\documentclass[journal]{IEEEtran}
\usepackage{amsmath,amsfonts,amssymb,euscript, graphicx, 
epsfig,
enumerate,float,afterpage, subfigure, ifthen}%

\usepackage{url}
\usepackage{hyperref}
\usepackage{breakurl}

\newtheorem{lem}{Lemma}

\newcommand{\defequiv}{\mbox{\raisebox{-.3ex}{$\overset{\vartriangle}{=}$}}}

\newcommand{\norm}[1]{||{#1}||}

\newcommand{\bv}[1]{{\boldsymbol{#1} }}
\newcommand{\script}[1]{{{\cal{#1} }}}

\begin{document}

\title
  {Wireless Peer-to-Peer Scheduling in Mobile Networks}
\author{Michael J. Neely \\ University of Southern California \\
\url{http://www-bcf.usc.edu/~mjneely}
\thanks{The author is with the  Electrical Engineering department at the University
of Southern California, Los Angeles, CA.}
\thanks{This material is supported in part  by one or more of:  the NSF Career grant CCF-0747525,   the 
Network Science Collaborative Technology Alliance sponsored
by the U.S. Army Research Laboratory W911NF-09-2-0053.}}

\markboth{}{Neely}

\maketitle

\begin{abstract}   
This paper considers peer-to-peer scheduling for a network with  
multiple wireless devices.  A subset of the devices are mobile users that desire
specific files.   Each user may already have certain popular files in 
its cache.  The remaining devices are access points that typically have access to 
a larger set of files. Users can download packets of their requested file from
an access point or from a nearby user. 
Our prior work optimizes peer scheduling in a general setting, but the resulting 
delay can be large when applied to mobile networks. 
This paper focuses on the mobile case, and 
develops a new algorithm that reduces delay by opportunistically 
grabbing packets from current neighbors. 
However, it treats a simpler model where each user desires a single file with 
infinite length.  An algorithm that provably optimizes throughput utility while incentivizing participation 
is developed for this case.  The algorithm extends as a simple heuristic 
in more general 
cases with finite file sizes and random active and idle periods. 
\end{abstract}

\section{Introduction} 

Consider a network with $N$ wireless devices.   Let $K$ of these devices be identified as 
\emph{users}, where $K \leq N$.  
The users 
can send packets to each other via direct peer-to-peer transmissions.  Each user has a certain
collection of popular files in its cache.   
The remaining $N-K$ devices are  \emph{access points}.
The access points are connected to a larger network, such as 
the internet, and 
hence typically have access to a larger
set of files.  While a general network may have users that desire to upload
packets to the access points, this paper focuses 
only on the user downloads.  Thus, throughout
this paper it is assumed that the access points only send packets to the users, but do not receive
packets.  In contrast, the users can both send and receive.
The network is mobile, and so the transmission options between access points and users, and between
user pairs, can change over time.

Each user only wishes to download, and does not naturally want to send any data. 
Users will only send data to each other if they agree to operate according to a  control algorithm that 
schedules such transmissions.  This paper assumes the users have already agreed to abide by 
the control algorithm, and thus focuses attention on altruistic network design that optimizes 
a global network utility function.  Nevertheless, 
this paper includes a tit-for-tat constraint, similar 
to the work in \cite{neely-p2p-infocom2011}, which is an effective mechanism for incentivizing participation.  
Further, while this paper develops a single algorithm for optimizing the entire 
network, this does not automatically
require the algorithm to be centralized.  Indeed, 
the resulting algorithm often has a distributed implementation.

The model of this paper applies to a variety of practical network situations.  For example, 
the access points can be wireless base stations in a future cellular network that 
allows both base-station-to-user transmissions as well as direct user-to-user transmissions. 
Alternatively, some of the access points can be smaller femtocell nodes. 
To increase the capacity of wireless systems, it is essential for future networks to enable 
such femtocell access and/or direct user-to-user transmission.   
This paper considers only 1-hop communication, so that
all downloads are received either directly from an access point or from another user.  
The possibility of a user acting as a multi-hop relay is not considered here. 

Our prior work \cite{neely-p2p-infocom2011} treats a more complex model where each user can actively
download multiple files at the same time, and where the arrival process of desired files at each user is random. 
A key challenge in this case is the complexity explosion associated with labeling each file according to the 
subset of other devices that already have it.  
This is solved in \cite{neely-p2p-infocom2011} 
by first observing the subset information of each newly arriving file, making an immediate 
decision about which device in this subset should transmit the desired packets, and 
then placing this request in a \emph{request queue} at that selected device. 
The devices are not required to transmit these packets immediately. Rather, they can satisfy 
the requests over time.   This procedure does not sacrifice optimality, and yields 
an algorithm with polynomial complexity. 
However, while this is effective for 
networks with static topology, it can result 
in significant delays in mobile networks.  That is because a device that is pre-selected
for transmission may not currently be in close proximity to the intended user, and/or may move
out of transmission range before transmission occurs.  

The current paper provides an alternative algorithm that reduces delay, particularly in the 
mobile case.  To do so, we use a simpler model that assumes each user desires only one file that has
infinite size.  This enables us to focus on scheduling to achieve optimally 
fair download rates for each user.   Rather than pre-selecting devices for eventual transmission, our
algorithm makes opportunistic packet transmission decisions from the set of current neighbors that have the 
desired file.  This also facilitates distributed implementation.  The infinite file size assumption is  
an approximation that is reasonable when file sizes are large, such as for video files.  
A heuristic extension to the case of finite file sizes is treated in Section \ref{section:finite}.
This heuristic is based on the optimality insights obtained from the infinite size case.  

Prior work on fair scheduling in mobile ad-hoc networks has considered token-based and
economics-based mechanisms to incentivize 
participation \cite{stimulate-cooperation1}\cite{crowcroft-incentives}\cite{neely-pricing-journal}. 
Incentives are also well studied in the peer-to-peer literature.
For example, algorithms in 
\cite{jun-p2pecon05}\cite{bharambe-infocom06}\cite{tamilmani04swift}\cite{piatek-nsdi07}
 track the number of uploads and downloads for each user, and give preferential treatment 
 to those who have helped others.  Algorithms based on tokens, markets, and peer reputations
 are considered in  \cite{liao-perf07}\cite{lian-iptps06}\cite{freedman-iptps08}.  Such algorithms
 are conceptually based on the simple ``tit-for-tat'' or ``treat-for-treat'' principle, where rewards 
 are given in direct proportion to the amount of self-sacrificial behavior at each user. 
Our current paper also considers a tit-for-tat mechanism.  However, a key difference is that 
it designs a \emph{tit-for-tat constraint} directly into the optimization problem
(similar to our prior work \cite{neely-p2p-infocom2011} that used a different network model). 
Remarkably, the 
solution of the optimization naturally results in an intuitive token-like procedure, where the number
 of tokens in a virtual queue at each user  
determines the \emph{reputation} of the user.   Peer-to-peer transmissions between user pairs 
are given preferential treatment according to the \emph{differential reputation} 
between the users. This is similar to the \emph{backpressure} principle for optimal network scheduling
\cite{tass-radio-nets}\cite{now}. However, the ``backpressure'' in the present paper is determined by 
token differentials in virtual reputation queues, rather than congestion differentials.

\section{Basic Model} 

Let  $\script{N}$ represent the set of devices, and $\script{K}$ represent the 
set of users, where $\script{K} \subseteq \script{N}$.   
Let $N$ and $K$ be the sizes of these sets. 
For simplicity in this basic model, 
assume that each user $k \in \script{K}$ wants a single file that consists of an infinite number
of fixed-length packets 
(this assumption is modified to treat finite file sizes in Section \ref{section:finite}). 
For each $k \in \script{K}$, define 
$\script{F}_k$ as the subset of devices in $\script{N}$ that have the file desired by user $k$.
The set $\script{F}_k$ can include both users and access points, and represents the set of 
devices that user $k$ can potentially receive packets from as it moves throughout the network. 
If each user $k$ only accepts downloads from a certain subset of devices that it identifies as its 
\emph{social group}, then $\script{F}_k$ can be viewed as 
the intersection of its social group
and the set of devices that have its file. 
We assume that $k \notin \script{F}_k$, 
so that no user wants a file that it already has.  Further, we assume 
$\script{F}_k$ is non-empty, so that at least
one device has the desired file.  This latter assumption is reasonable if the network has 
one or more access points with high speed 
internet connections.  Such access points can be modeled 
as having all desired files. 

The network operates in slotted time with time slots $t \in \{0, 1, 2, \ldots\}$. Every slot $t$, the network 
makes \emph{transmission actions}.  Let $(\mu_{nk}(t))$ be the matrix 
of transmission actions chosen on slot $t$, where each entry $\mu_{nk}(t)$ is the number of packets
that device $n \in \script{N}$ transmits to user $k \in \script{K}$.  The set of all possible matrix options to 
choose from 
is determined by the current \emph{topology state} of the network, as in \cite{now}.  Specifically, 
let $\omega(t)$ represent the topology state on slot $t$, being a vector of parameters
that affect transmission, 
such as current device locations and/or channel conditions. 
Assume $\omega(t)$ takes values in an abstract set $\Omega$, possibly being an infinite 
set. 
The $\omega(t)$ process is assumed
to be ergodic.  In the case when $\Omega$ is finite or countably infinite, the 
steady state probabilities are represented by $\pi(\omega) = Pr[\omega(t) = \omega]$ for all $\omega \in \Omega$. 
Else, the steady state probabilities are represented by an appropriate probability density function. 
These probabilities are not necessarily known by the network controller.  Extensions to the case when 
$\omega(t)$ is non-ergodic are explored in Sections \ref{section:finite} and \ref{section:simulation}.  

For each $\omega \in \Omega$, define $\script{R}(\omega)$
as the set of all transmission matrices $(\mu_{nk})$ that are possible when $\omega(t) =\omega$.
The exact structure of $\script{R}(\omega)$ 
depends on the particular physical characteristics and interference properties of the network.   
One 
may choose the $\script{R}(\omega)$ sets to constrain the transmission variables $\mu_{nk}$
to take integer values, although this is not required in our analysis.  
We assume only 
that each set $\script{R}(\omega)$ has the following basic properties: 
\begin{itemize} 
\item Every matrix $(\mu_{nk}) \in \script{R}(\omega)$ has  non-negative entries.

\item If $(\mu_{nk}) \in \script{R}(\omega)$, then $(\tilde{\mu}_{nk}) \in \script{R}(\omega)$, where
$(\tilde{\mu}_{nk})$ is any matrix formed by setting one or more entries of $(\mu_{nk})$ to zero. 

\item Every matrix $(\mu_{nk}) \in \script{R}(\omega)$ must  
satisfy the constraint 
 $0 \leq \sum_{n\in\script{N}} \mu_{nk} \leq x_{k}^{max}$ for all $k \in \script{K}$, where $x_{k}^{max}$
is a given bound on the number of packets that can be delivered to user $k$ on 
one slot, regardless of  $\omega$. 
\end{itemize}

\subsection{Example Network Structure} \label{section:example-network-structure} 

This section presents an example network model that fits into the above framework. 
The network region is divided into $C$ non-overlapping subcells.  Let $c_n(t)$ be the 
current subcell of device $n$, so that $c_n(t) \in \{1, \ldots, C\}$. Let
$\bv{c}(t) = (c_1(t), \ldots, c_N(t))$ 
be the vector of device locations on slot $t$.  
The access points have fixed locations, 
while the mobile users can change subcells from slot to slot.  Let $\bv{S}(t) = (S_{nk}(t))$ 
be a channel state matrix, where $S_{nk}(t)$ is the number
of packets that device $n$ can transmit to device $k$ on slot $t$, provided that there are no 
competing transmissions (as defined below). 
The value $S_{nk}(t)$ can depend 
on the $\bv{c}(t)$ location
vector.  
Let the topology state $\omega(t)$ be given by
$\omega(t) = (\bv{c}(t), \bv{S}(t))$. 
Assume access point transmissions are orthogonal from all 
other access point transmissions and from all 
peer-to-peer user transmissions. For each $\omega(t)$, define $\script{R}(\omega(t))$ as
the set of all $(\mu_{nk}(t))$ matrices with entries that satisfy: 
\begin{itemize} 
\item  $\mu_{nk}(t) \in \{0, S_{nk}(t)\}$ for all $n \in \script{N}$ and $k \in \script{K}$.
\item Users can only transmit to other users currently in their same subcell, so that 
$\mu_{nk}(t)=0$ whenever $n \in \script{K}$, $k\in\script{K}$, and $c_n(t) \neq c_k(t)$. 
\item At most one user-to-user transmission can take  place per subcell on a given slot. 
\item Each access point can send to at most one user per slot. 
\end{itemize} 
This particular structure is useful because it allows user transmissions to be separately scheduled
in each subcell, and access point transmissions to be scheduled separately from all other decisions. 
The sets $\script{R}(\omega(t))$ can be defined differently 
for more 
sophisticated interference models.  For example, the transmission rate of an 
access point can depend on whether or not neighboring access points are scheduled for transmission.

\subsection{Optimization Objective} 

For each $a \in \script{N}$ and $b \in \script{K}$, define $f_{ab}(t)$ to be $1$ if, on slot $t$, 
device $a$ has the 
file requested by user $b$, and $0$ otherwise: 
\begin{equation} \label{eq:fab}  f_{ab}(t) \defequiv  \left\{ \begin{array}{ll}
                          1 &\mbox{ if $a \in \script{F}_b$} \\
                             0  & \mbox{ otherwise} 
                            \end{array}
                                 \right. 
\end{equation} 
Because each user desires a single infinite size file, 
the  $f_{ab}(t)$ values do not change with time.  However, we use the ``$(t)$'' notation because it 
facilitates the extension to more general cases in Section \ref{section:finite}, where the sets $\script{F}_b$ in the right-hand-side of \eqref{eq:fab} are
extended to $\script{F}_b(t)$. 
Every slot $t$, the network controller observes $\omega(t)$ and chooses $(\mu_{nk}(t)) \in \script{R}(\omega(t))$. 
For each $k \in \script{K}$, define $x_k(t)$ as the total number
 of packets that user $k$ receives from others on slot $t$,  and define $y_k(t)$ as the total number of packets that
 user $k$ delivers to others on slot $t$: 
 \begin{eqnarray} 
 \mbox{$x_k(t) \defequiv \sum_{a\in \script{N}} \mu_{ak}(t)f_{ak}(t)$} \label{eq:xk}  \\
  \mbox{$y_k(t) \defequiv \sum_{b \in \script{K}} \mu_{kb}(t)f_{kb}(t)$} \label{eq:yk}
\end{eqnarray} 
The multiplication $\mu_{ab}(t)f_{ab}(t)$ in \eqref{eq:xk} and \eqref{eq:yk} formally ensures that
user $b$ can only receive a packet from another device that has the file it is requesting. 

For a given control algorithm, let $\overline{x}_k$ and $\overline{y}_k$
represent the time averages of the $x_k(t)$ and $y_k(t)$ processes for all $k \in \script{K}$: 
\begin{eqnarray*}
\mbox{$\overline{x}_{k} = \lim_{t\rightarrow\infty} \frac{1}{t}\sum_{\tau=0}^{t-1} x_{k}(\tau) \: \: , \: \: 
\overline{y}_k = \lim_{t\rightarrow\infty} \frac{1}{t}\sum_{\tau=0}^{t-1} y_k(\tau)$}
\end{eqnarray*}
These limits are temporarily assumed to exist.\footnote{This is only to simplify exposition of 
the optimization 
goal. The analysis in later sections does not a-priori 
assume the limits exist.}  The value $\overline{x}_k$ is the time average download rate of user $k$, 
and $\overline{y}_k$ is the time average upload rate. 
The goal is to develop a control algorithm that solves the following stochastic network 
optimization problem:\footnote{Note that the trivial all-zero
solution is always feasible.} 
\begin{eqnarray} 
\mbox{Maximize:} & \sum_{k\in\script{K}} \phi_k(\overline{x}_k) \label{eq:p0} \\
\mbox{Subject to:} & \alpha_k \overline{x}_k \leq  \beta_k + \overline{y}_k \: \: \forall k \in \script{K} \label{eq:p1}  \\
& (\mu_{nk}(t)) \in \script{R}(\omega(t)) \: \: \forall t \in \{0, 1, 2, \ldots\}  \label{eq:p2} 
\end{eqnarray} 
 where for each $k \in \script{K}$, $\phi_k(x)$ 
 are given concave functions and $\alpha_k, \beta_k$ are 
 given non-negative weights. 
 The value 
 $\phi_k(\overline{x}_k)$ represents the \emph{utility} associated with user $k$ downloading at 
 rate $\overline{x}_k$.   The constraints \eqref{eq:p1} are the \emph{tit-for-tat} constraints from 
 \cite{neely-p2p-infocom2011}.  These constraints incentivize
 participation.  They allow a ``free'' download rate of $\beta_k/\alpha_k$.  Users can only receive rates beyond
 this value in proportion to the  
 rate at which they help others.  
 The value $\beta_k$ can be set to zero to remove
 the ``free'' rate, and the value $\alpha_k$ can be set to $0$ to remove the tit-for-tat constraint for user
 $k$.  
 Choosing larger values of $\alpha_k$ (typically in the range $0 \leq \alpha_k \leq 1$) 
 leads to more stringent requirements about helping others. 
 These tit-for-tat constraints restrict the system operation and thus can affect overall network utility. 
 Removing these constraints by setting $\alpha_k=0$ for all $k$ leads to the largest network utility, 
 but does not embed any participation incentives into the optimization problem.

 The functions $\phi_k(x)$ are assumed to be 
 concave, continuous, and non-decreasing over the interval $x \geq 0$.  
 They are not required to be differentiable.
 For  example, they can be piecewise linear, such as
 $\phi_k(x) = \min[x, \theta_k]$, 
 where $\theta_k$ is a given constant rate desired by user $k$. 
 Alternatively, one can choose $\phi_k(x) = \ln(x)$ for each $k\in\script{K}$, which leads to 
 the well known \emph{proportional fairness utility} \cite{kelly-charging}. 
 
 One may want to modify the problem \eqref{eq:p0}-\eqref{eq:p2} by specifying separate utility functions
 for the user-to-user download rates and the access-point-to-user download rates.  This is possible
 by creating two  ``virtual users'' $m_1(k)$ and $m_2(k)$  for each actual user $k \in \script{K}$.
 Channel conditions
 for the virtual users $m_1(k), m_2(k)$ are defined to be the 
 same as for the actual user $k$, 
 with the exception that virtual user $m_1(k)$ is restricted to receive only from other users, 
 while virtual user $m_2(k)$ is restricted to receive only from the access points.

 \section{The Dynamic Algorithm} 
 
 The problem \eqref{eq:p0}-\eqref{eq:p2} is solved via 
 the stochastic network optimization theory of 
 \cite{now}\cite{sno-text}. First note that problem \eqref{eq:p0}-\eqref{eq:p2} 
 is equivalent to the following problem that uses
 \emph{auxiliary variables} $\gamma_k(t)$: 
 \begin{eqnarray} 
 \mbox{Maximize:} & \sum_{k\in\script{K}} \phi_k(\overline{\gamma}_k) \label{eq:q0}  \\
 \mbox{Subject to:} & \alpha_k\overline{x}_k \leq \beta_k + \overline{y}_k \: \: \forall k \in \script{K} \label{eq:q1} \\
 & \overline{\gamma}_k \leq \overline{x}_k \: \: \forall k \in \script{K} \label{eq:q2} \\
 & (\mu_{nk}(t)) \in \script{R}(\omega(t)) \: \: \forall t \in \{0, 1, 2, \ldots\} \label{eq:q3} \\
 & 0 \leq \gamma_k(t) \leq x_k^{max} \: \: \forall t \in \{0, 1, 2, \ldots\}  \label{eq:q4} 
 \end{eqnarray} 
 The auxiliary variables $\gamma_k(t)$ act as proxies for the actual download
 variables $x_k(t)$. This is useful for separating the nonlinear 
 utility optimization from the 
 network transmission decisions.   
 
 It can be shown that the optimal utility value is the same for both problems
  \eqref{eq:p0}-\eqref{eq:p2} and 
 \eqref{eq:q0}-\eqref{eq:q4} \cite{sno-text}.  Let $\phi^*$ denote this optimal utility.
 Now consider an algorithm that
 solves the problem \eqref{eq:q0}-\eqref{eq:q4}.  Let $(\mu_{nk}(t))$,  $(\gamma_k(t))$ be the decisions 
made over time, and let $\overline{x}_k$, $\overline{y}_k$, $\overline{\gamma}_k$
be the corresponding time averages, 
all of which satisfy the constraints of the problem \eqref{eq:q0}-\eqref{eq:q4}.   Note that these
constraints 
 include all of the desired constraints of the original problem
\eqref{eq:p0}-\eqref{eq:p2}. 
Then: 
\begin{eqnarray}
\phi^* &=& \mbox{$\sum_{k\in\script{K}} \phi_k(\overline{\gamma}_k)$}   \label{eq:intuitive-1} \\
&\leq& \mbox{$\sum_{k\in\script{K}} \phi_k(\overline{x}_k)$} \label{eq:intuitive-2} \\
&\leq& \phi^* \label{eq:intuitive-3} 
\end{eqnarray}
where \eqref{eq:intuitive-1} holds because this algorithm achieves the optimal utility 
$\phi^*$, \eqref{eq:intuitive-2} holds because this algorithm must yield time averages
that satisfy $\overline{\gamma}_k \leq \overline{x}_k$ for all $k\in\script{K}$, and \eqref{eq:intuitive-3}
holds because the transmission decisions of the algorithm satisfy all desired 
constraints of the original problem, and thus produce $\overline{x}_k$ values that give a 
utility that is less than or equal to the optimal utility of the original problem (which is also $\phi^*$).   
It follows that any algorithm
that is optimal for \eqref{eq:q0}-\eqref{eq:q4} makes decisions that are also optimal for the 
original problem.

 \subsection{Virtual Queues} 
 
 To facilitate satisfaction of the tit-for-tat constraints \eqref{eq:q1}, for each $k \in \script{K}$ define 
 a \emph{virtual queue} $H_k(t)$, with dynamics: 
 \begin{equation} \label{eq:h-update} 
 H_k(t+1) = \max\left[H_k(t) + \alpha_kx_k(t) - \beta_k  - y_k(t), 0\right] 
 \end{equation} 
 where  $x_k(t)$, $y_k(t)$ are defined in \eqref{eq:xk}-\eqref{eq:yk}. 
 The intuition is that $\alpha_kx_k(t)$ can be viewed as  the ``arrivals'' on slot $t$, and 
 $\beta_k +y_k(t)$ can be viewed as  the 
 ``offered service'' on slot $t$.  Stabilizing queue $H_k(t)$ ensures the time average of the 
 ``arrivals'' is less than or equal to the time average of the ``service,'' which ensures constraints
 \eqref{eq:q1}.

 Similarly, to satisfy the constraints \eqref{eq:q2}, for each $k \in \script{K}$ define another
 virtual queue $Q_k(t)$ with dynamics: 
 \begin{equation} \label{eq:q-update} 
 Q_k(t+1) = \max[Q_k(t) + \gamma_k(t) -  x_k(t), 0] 
 \end{equation} 
 The update \eqref{eq:q-update} can be interpreted as
 a queueing equation where $\gamma_k(t)$ is the amount of data requested by 
 user $k$ on slot $t$, and $x_k(t)$ is the amount of service. 
 Stabilizing $Q_k(t)$ ensures
 $\overline{\gamma}_k \leq \overline{x}_k$. 
   
  \subsection{The Drift-Plus-Penalty Algorithm} 
  
  Define the following quadratic function $L(t)$: 
  \[ L(t) \defequiv \mbox{$\frac{1}{2}\sum_{k\in\script{K}} [Q_k(t)^2 + H_k(t)^2]$} \]
  Intuitively, 
  taking actions to push $L(t)$ down tends to maintain stability of all queues. 
  Define $\Delta(t)$ as the drift on slot $t$: 
  \[ \Delta(t) \defequiv L(t+1) - L(t) \]
  Let $\bv{\Theta}(t) = (Q_k(t), H_k(t))|_{k\in\script{K}}$ be the vector of all virtual queue values 
  on slot $t$.  
  The algorithm is designed to observe the queues and the current $\omega(t)$ on each slot $t$, 
  and to  then  choose $(\mu_{nk}(t)) \in \script{R}(\omega(t))$ and $\gamma_k(t)$ subject to 
$0 \leq \gamma_k(t) \leq x_k^{max}$ to 
  minimize a bound on the following \emph{drift-plus-penalty expression} \cite{sno-text}: 
  \[ \Delta(t) - V\mbox{$\sum_{k\in\script{K}} \phi_k(\gamma_k(t))$} \]
  where $V$ is a non-negative weight that affects a performance bound.  Intuitively, the value of $V$
  affects the extent to which our control action on slot $t$ 
  emphasizes utility optimization in comparison to drift minimization. 
  
  \begin{lem} \label{lem:dpp}  Under any control algorithm, we have: 
  \begin{eqnarray}
  \Delta(t) - V\mbox{$\sum_{k\in\script{K}} \phi_k(\gamma_k(t))$} \leq B(t) - V\mbox{$\sum_{k\in\script{K}}$}\phi_k(\gamma_k(t))\nonumber \\
  + \mbox{$\sum_{k\in\script{K}} H_k(t)[\alpha_kx_k(t) - \beta_k - y_k(t)]$} \nonumber \\
  + \mbox{$\sum_{k\in\script{K}} Q_k(t)[\gamma_k(t) -x_k(t)]$} \label{eq:dpp} 
  \end{eqnarray}
  where $B(t)$ is defined: 
  \begin{eqnarray*}
  B(t) &\defequiv& \mbox{$\frac{1}{2}\sum_{k\in\script{K}} (\alpha_kx_k(t) - \beta_k - y_k(t))^2$} \\
  && + \mbox{$\frac{1}{2}\sum_{k\in\script{K}} (\gamma_k(t)-x_k(t))^2$}
  \end{eqnarray*}
  \end{lem} 
  
  \begin{proof} 
  Squaring \eqref{eq:h-update} and using $\max[y,0]^2 \leq y^2$ for any real number $y$ yields: 
  \begin{eqnarray*}
  H_k(t+1)^2 &\leq& H_k(t)^2 + (\alpha_kx_k(t) - \beta_k - y_k(t))^2 \\
  && 2H_k(t)(\alpha_kx_k(t) - \beta_k - y_k(t))
  \end{eqnarray*}
  Similarly, squaring \eqref{eq:q-update} gives: 
  \begin{eqnarray*}
  Q_k(t+1)^2 &\leq& Q_k(t)^2 + (\gamma_k(t) - x_k(t))^2 \\
  && + 2Q_k(t)(\gamma_k(t) - x_k(t)) 
  \end{eqnarray*}
  Summing over $k \in \script{K}$ and dividing by 2 yields the result. 
  \end{proof} 
  
  The value of $B(t)$ can be upper bounded by a finite constant $B$ every slot, where $B$ depends on
  the maximum possible values that $\mu_{nk}(t)$ and $\gamma_k(t)$  can take. 
  The algorithm below is defined by observing the queue states and $\omega(t)$ every slot $t$, 
  and choosing actions to minimize the last three terms on the right-hand-side of \eqref{eq:dpp} (not including
  the first term $B(t)$), given these
  observed quantities. Specifically, every slot $t$: 
  \begin{itemize} 
  \item ($\gamma_k(t)$ decisions) Each user $k \in \script{K}$ observes $Q_k(t)$ and chooses $\gamma_k(t)$
  to solve: 
  \begin{eqnarray*}
  \mbox{Maximize:} & V\phi_k(\gamma_k(t)) - Q_k(t)\gamma_k(t) \\
  \mbox{Subject to:} & 0 \leq \gamma_k(t) \leq x_k^{max} 
  \end{eqnarray*}
  
  \item ($\mu_{nk}(t)$ decisions) The network controller observes all queues $(\bv{Q}(t), \bv{H}(t))$ and
  the topology state $\omega(t)$ on slot $t$, and chooses matrix $(\mu_{nk}(t)) \in \script{R}(\omega(t))$ to 
  maximize the following expression: 
  \begin{eqnarray}
  \mbox{$\sum_{n\in \script{N}}\sum_{k\in\script{K}}\mu_{nk}(t)f_{nk}(t)W_{nk}(t)$} \label{eq:transmission-alg}  
  \end{eqnarray}
  where weights $W_{nk}(t)$ are defined: 
  \[ W_{nk}(t) \defequiv \left[Q_k(t) + 1_{\{n\in\script{K}\}}H_n(t) - \alpha_kH_k(t) \right] \]
  where $1_{\{n\in\script{K}\}}$ is an indicator function that is 
  $1$ if device $n$ is a user, and $0$ if device $n$ is an access point. 
  
  \item (Queue updates) Update virtual queues $H_k(t)$ and $Q_k(t)$ 
  for all $k \in \script{K}$ via \eqref{eq:h-update} and \eqref{eq:q-update}.
  \end{itemize}   
  
  The $\gamma_k(t)$ decisions can be viewed as flow control actions that restrict the amount of 
  data requested from user $k$ on each slot.  They are made separately
  at  each user $k$. 
  The $(\mu_{nk}(t))$ decisions are transmission
  actions made at the network layer.  Examples are given below.

  \subsection{Example Flow Control Decisions} 
  
  Suppose we have the following piecewise linear utility functions for all $k \in \script{K}$: 
  \begin{equation} \label{eq:piecewise-linear} 
   \phi_k(x_k) = \nu_k\max[x_k, \theta_k]
  \end{equation}  
 where $\nu_k$ are given positive values that act as priority weights for the users,
 and $\theta_k$ are are given positive values that represent the maximum desired 
 communication rate for each user.  Assume that $\theta_k \leq x_k^{max}$ for all $k \in \script{K}$. 
 Then the algorithm in the previous section chooses $\gamma_k(t)$ for each user $k \in \script{K}$ as: 
 \begin{equation*} 
  \gamma_k(t) =  \left\{ \begin{array}{ll}
                          \theta_k  &\mbox{ if $Q_k(t) \leq V\nu_k$} \\
                             0  & \mbox{ otherwise} 
                            \end{array}
                                 \right.
  \end{equation*}
  
  Alternatively, 
  suppose we have the following strictly concave utility functions for all $k \in \script{K}$: 
  \begin{equation} \label{eq:prop-fair-approx}
  \phi_k(x_k) = \ln(1+ \nu_kx_k)
  \end{equation}  
 Then the $\gamma_k(t)$ decisions are: 
  \begin{eqnarray}
   \gamma_k(t) = 
                         \left[\frac{V}{Q_k(t)} -\frac{1}{\nu_k}\right]_0^{x_k^{max}} \label{eq:prop-fair-approx-gamma}
 \end{eqnarray}
 where the operation $[y]_0^{b}$ is equal to $y$ if $0 \leq y \leq b$, $0$ if $y < 0$, and $b$ if $y>b$.  
 These utility functions can be viewed as an accurate approximation of the proportionally fair utility function
 if we use $\nu_k = \nu$ for all $k$, for a large value of $\nu$. 
 
 Finally, using the proportionally fair utilities 
 $\phi_x(x_k) = \ln(x_k)$ for all $k \in \script{K}$ leads to 
 $\gamma_k(t) = [V/Q_k(t)]_0^{x_k^{max}}$, which is indeed the same as \eqref{eq:prop-fair-approx-gamma}
 in the limit as $\nu_k\rightarrow\infty$.

 \subsection{Example Transmission Decisions}
 
 Suppose the network has the special structure specified in Section \ref{section:example-network-structure}. 
 Let $\script{A}$ be the set of access points. The set $\script{A}$ is disjoint from the user set $\script{K}$,
 and $\script{A} \cup \script{K} = \script{N}$.  
 Let $\script{K}_a(t)$ be the set of users within reach of access point
 $a$ on slot $t$.  Then each access point $a \in \script{A}$ observes channels $S_{ak}(t)$ and 
 queues $Q_k(t)$, $H_k(t)$ for all users $k \in \script{K}_a(t)$ and chooses to serve the single
 user in $\script{K}_a(t)$ with the largest non-negative 
 value of   $f_{ak}(t)S_{ak}(t)[Q_k(t) - \alpha_kH_k(t)]$ (breaking ties arbitrarily), 
 and chooses no users if this value is negative for 
 all $k \in \script{K}_a(t)$. 
 
 Further, the user pairs in each subcell $c \in \{1, \ldots, C\}$ are observed.  Amongst all users
 in a given subcell, the ordered
 user pair $(a,k)$ with the largest non-negative value of $f_{ak}(t)S_{ak}(t)[Q_k(t) + H_a(t) - \alpha_kH_k(t)]$
  is selected for peer-to-peer transmission in that subcell (breaking ties arbitrarily). 
  No peer-to-peer transmission occurs in the
  subcell if this value is negative for all user pairs. 
 This can be coordinated by a controller in each subcell, 
  or can be done by having each user first transmit control information from which the winning user pair
  is determined.  

  \section{Algorithm Performance} 
  
 The algorithm fits into the stochastic network optimization framework of \cite{sno-text}, and hence
 it satisfies all desired constraints of the original problem \eqref{eq:p0}-\eqref{eq:p2} (as well as the 
 transformed problem \eqref{eq:q0}-\eqref{eq:q4}), with an overall utility value that differs from 
 the optimal $\phi^*$ by at most $O(1/V)$, which can be made arbitrarily small as the parameter $V$ is increased. 
 However, the $V$ parameter directly affects the size of the virtual queues, which affects convergence time
 of the algorithm to this utility.  
 Here we show that if the utility functions and network transmission sets have additional structure, the 
 virtual queues $Q_k(t)$ and $H_k(t)$ can be deterministically bounded for all 
 time by a constant that is proportional to $V$. 
 
 \subsection{Bound on Data Queues $Q_k(t)$} 
 
 Suppose each utility function $\phi_k(t)$ has right-derivatives that are bounded by a finite 
 constant $\nu_k>0$ over the interval $0 \leq x_k \leq x_{k}^{max}$.  For example, this holds for the 
piecewise linear utility functions \eqref{eq:piecewise-linear} and the strictly concave 
 utility functions
 \eqref{eq:prop-fair-approx}, with the $\nu_k$ parameters specified there indeed being the maximum 
 right derivatives.  
 
 \begin{lem} \label{lem:q-bound} If utility function $\phi_k(x)$ has maximum right derivative $\nu_k>0$, then:
 \[ 0 \leq Q_k(t) \leq V\nu_k + x_k^{max} \: \: \forall t \in \{0, 1, 2, \ldots\} \] 
 provided that this inequality holds for $Q_k(0)$. 
 \end{lem} 
 
 \begin{proof}
 Assume that $Q_k(t)\leq V\nu_k + x_k^{max}$ for slot $t$ (it holds by assumption on slot $t=0$).  We prove it also holds for slot
 $t+1$.  First consider the case $Q_k(t)\leq V\nu_k$.  From the queue update equation \eqref{eq:q-update}, 
 we see that this queue can increase by at most $x_k^{max}$ on each slot, and so 
 we have $Q_k(t+1) \leq V\nu_k + x_k^{max}$, proving the result for this case. 
 
 Now consider the case $Q_k(t) > V\nu_k$.  On slot $t$, 
 the algorithm chooses $\gamma_k(t) \in [0, x_k^{max}]$ to maximize the expression: 
\[ V\phi_k(\gamma_k(t)) - Q_k(t)\gamma_k(t) \]
However, for any $\gamma_k(t) \geq 0$ we have: 
\begin{eqnarray*}
&& \hspace{-.5in} V\phi_k(\gamma_k(t)) - Q_k(t) \gamma_k(t) \\
&\leq& V\phi_k(0) + V\nu_k\gamma_k(t) - Q_k(t)\gamma_k(t) \\
&=& V\phi_k(0) + \gamma_k(t)[V\nu_k - Q_k(t)] \\
&\leq& V\phi_k(0) 
\end{eqnarray*}
where equality holds if and only if $\gamma_k=0$ (recall that $[V\nu_k - Q_k(t)] < 0$). 
It follows that the algorithm must choose
$\gamma_k(t)=0$, and so $Q_k(t)$ cannot increase on this slot.  That is: 
\[ Q_k(t+1) \leq Q_k(t) \leq V\nu_k + x_{k}^{max} \]
 \end{proof}

\subsection{Bound on Reputation Queues $H_k(t)$}

 The $H_k(t)$ processes act as \emph{reputation queues} for each user $k \in \script{K}$, being low if 
 user $k$ has a good reputation for helping others, and high otherwise (see queue dynamics in 
 \eqref{eq:h-update}).   These reputations directly
 affect the transmission decisions via the weights $W_{nk}(t)$ in 
 \eqref{eq:transmission-alg}.   To see this, define $\script{A}$ as the set of access points. 
 First consider the weight 
 seen by an access point $a \in \script{A}$ for user $k$ on slot $t$: 
 \[ Q_k(t) - \alpha_k H_k(t) \]
 This weight is large if $H_k(t)$ is small
 (meaning user $k$ has a good reputation).  
The next lemma shows that if utility functions have bounded right-derivatives $\nu_k$, then
access points will  refuse to send to any user if its $H_k(t)$ value exceeds a threshold. 

\begin{lem} If $\phi_k(t)$ has maximum right-derivative $\nu_k>0$,  and if initial queue backlog
satisfies $0 \leq Q_k(0) \leq V\nu_k + x_k^{max}$, then no access point 
will send to user $k$ on a given slot $t$ if $H_k(t) > \frac{1}{\alpha_k}[V\nu_k + x_k^{max}]$. 
\end{lem} 

\begin{proof} 
Lemma \ref{lem:q-bound}
 ensures $Q_k(t) \leq V\nu_k + x_k^{max}$ for all $t$.  It follows that if $H_k(t) > \frac{1}{\alpha_k}[V\nu_k + x_x^{max}]$, then the weight seen by an access point $a\in\script{A}$ for user $k$ satisfies: 
 \begin{eqnarray*}
 Q_k(t) - \alpha_kH_k(t) &\leq& V\nu_k + x_k^{max} - \alpha_kH_k(t) \\
 &<& 0 
 \end{eqnarray*}
 Because the weight is negative,  the max-weight functional 
 \eqref{eq:transmission-alg} is maximized by choosing $\mu_{ak}(t)=0$, so that access point  
 $a$ will not send data to user $k$ on slot $t$. 
 \end{proof}

 Now suppose user $u \in \script{K}$ considers transmitting to another user $k \in \script{K}$.  
 User $u$ sees
 the weight: 
 \[ Q_k(t) + H_u(t) - \alpha_kH_k(t) \]
 The value $H_u(t) -\alpha_kH_k(t)$ can be viewed as a \emph{differential reputation}. 
 We again see that a relatively low value of $H_k(t)$ improves the weights for user $k$.
  The next lemma shows that all queues $H_k(t)$ are deterministically bounded. 
  For simplicity, we state the lemma under the assumption that all initial queue backlogs
  are zero. 
 
 \begin{lem} \label{lem:h-bound} If all utility functions have right-derivatives bounded by 
 finite constants $\nu_k>0$, if $\beta_k>0$ for all $k \in \script{K}$, and if initial backlog satisfies
$Q_k(0) = H_k(0) =0$ for all $k \in \script{K}$, then there are finite constants $C_1$ and $C_2$,
both independent of $V$, such that: 
\begin{eqnarray*}
\norm{\bv{\Theta}(t)} \leq C_1 + C_2V \: \: \forall t \in \{0, 1, 2, \ldots\} 
\end{eqnarray*}
where $\norm{\bv{\Theta}(t)}$ is defined: 
\[ \norm{\bv{\Theta}(t)} \leq \mbox{$\sqrt{\sum_{k\in\script{K}}H_k(t)^2 + \sum_{k\in\script{K}}Q_k(t)^2}$} \]
 \end{lem} 
 
The constants $C_1$ and $C_2$ are given in the proof.

 \begin{proof} 
 See Appendix.
 \end{proof}

 \subsection{Constraint Satisfaction} 
 
 The deterministic queue bounds derived in the previous subsections ensure that 
 the time average tit-for-tat constraints \eqref{eq:q1}, as well as the auxiliary variable 
 constraints \eqref{eq:q2}, are satisfied on every sample path, regardless of whether or not
 the topology state process $\omega(t)$ is ergodic. This is shown in the next lemma. 
 Assume there are constants $Q_k^{max}$ and $H_k^{max}$ such that:  
 \begin{eqnarray} 
 0 \leq Q_k(t) \leq Q_k^{max} \: \: \forall t \in \{0, 1, 2, \ldots\} \label{eq:qmax}  \\
 0 \leq H_k(t) \leq H_k^{max } \: \: \forall t \in \{0, 1, 2, \ldots\}  \label{eq:hmax}
 \end{eqnarray} 
 Recall that Lemmas \ref{lem:q-bound} and \ref{lem:h-bound} ensure this holds 
 for $Q_k^{max} = V\nu_k + x_k^{max}$ and $H_k^{max} = C_1 + C_2V$ whenever queue backlogs
 are initially 0. 
 
 \begin{lem} If \eqref{eq:qmax} and \eqref{eq:hmax} hold, then for any slot $t \in \{0, 1, 2, \ldots\}$ 
 and any positive integer $T$: 
 
 (a) The tit-for-tat behavior over the interval $\tau \in \{t, \ldots, t+T-1\}$ satisfies the following
 for all $k \in \script{K}$: 
 \begin{eqnarray}
 \frac{1}{T}\sum_{\tau=t}^{t+T-1}\left[\alpha_kx_k(\tau) - \beta_k -  y_k(\tau)\right] \leq \frac{H_k^{max}}{T}  \label{eq:h-timeav} 
 \end{eqnarray}
 
 (b) The auxiliary variable constraints satisfy for all $k \in \script{K}$: 
 \begin{equation} \label{eq:q-timeav} 
  \frac{1}{T}\sum_{\tau=t}^{t+T-1}\left[\gamma_k(\tau) - x_k(\tau)\right] \leq \frac{Q_k^{max}}{T} 
  \end{equation} 
 \end{lem} 
 
 \begin{proof}
From \eqref{eq:h-update} we have for any slot $\tau \in \{0, 1, 2, \ldots\}$:  
\begin{eqnarray*}
H_k(\tau+1) \geq H_k(\tau) + \alpha_kx_k(\tau) - \beta_k - y_k(\tau)
\end{eqnarray*}
Summing the above over $\tau \in \{t, t+1, \ldots, t+T-1\}$ gives: 
\begin{eqnarray*}
H_k(t+T) - H_k(t) \geq \sum_{\tau=t}^{t+T-1}[\alpha_kx_k(\tau) - \beta_k - y_k(\tau)]
\end{eqnarray*}
However: 
\[ H_k(t+T) - H_k(t) \leq H_k^{max} \]
Combining the above two inequalities and dividing by $T$ yields \eqref{eq:h-timeav}.  The inequality
\eqref{eq:q-timeav} is proved similarly. 
 \end{proof} 
 
 Thus, even if the limits on the left-hand-sides of the inequalities 
 \eqref{eq:h-timeav} and \eqref{eq:q-timeav} 
 do not converge as $T\rightarrow\infty$ (possibly due to non-ergodic mobility),  
 their $\limsup$ 
 must be less than or equal to $0$.

\section{Extension to Finite File Sizes} \label{section:finite} 

Now suppose the files requested by users have finite sizes.  Suppose each user requests
at most one file at a time.  We say a user is in the \emph{active} state if it is requesting a file, and
is in the \emph{idle} state if it does not have any file requests.  
For each  $k \in \script{K}$, define
$A_k(t)$ to be 1 if user $k$ is active on slot $t$, and $0$ else.  If $A_k(t) =1$, define $\script{F}_k(t)$ as the 
set of devices in $\script{N}$ that have the currently requested file of user $k$.   Define $\script{F}_k(t)$ to be 
the empty set $\{\}$ if the user is not active on slot $t$.  Define $D_k(t)$ as the number
of additional required packets for user $k$ to complete its file request (where $D_k(t) >0$ if and only if $A_k(t)=1$). 
When an active 
user finishes downloading all packets of its requested file on some slot $t$, it goes to the idle state, so that 
$A_k(t+1)=D_k(t+1)=0$, and $\script{F}_k(t+1) = \{\}$.   

We can naturally extend the algorithm developed in the previous section to this case (although in this more
general case the algorithm is a heuristic).  The only 
change is that $\script{F}_b$ in \eqref{eq:fab}  is changed to $\script{F}_b(t)$. 
The algorithm proceeds exactly
as before, still updating queues and making all $\gamma_k(t)$ and $\mu_{ab}(t)$ decisions the same
way for all users on every slot, regardless of whether users are active or idle.  Of course, the $f_{ab}(t)$
parameters in \eqref{eq:transmission-alg} and in the receive and send equations 
\eqref{eq:xk}, \eqref{eq:yk} will remove any transmission link $(a,b)$ from consideration if user $b$
is idle.  However, idle users can still participate in data delivery to other users. 
If one wants to model a user
$k$ that wishes to turn its peer-delivery functionality ``off,'' one can use a modified 
set of transmission options that sets all outgoing links from user $k$ to have rate $0$. 
Intuitively, the algorithm will behave well, with performance close to that suggested by the infinite file size assumption, when file sizes are large.

\section{Simulation} \label{section:simulation} 

We simulate the algorithm on the cell-partitioned network structure of 
Section \ref{section:example-network-structure}, with $K=50$ users and  a single base station as access point. 
The users move according to a Markov random walk on a $4 \times 4$ grid with 16 subcells.  We use
utility functions $\phi_k(x)$ given by \eqref{eq:prop-fair-approx} with $\nu_k=1$.  Each cell can support at most 
one user-to-user packet transmission per slot. The base station can transmit to at most one user $k$ per slot,
with transmission rate $S_k(t)$ that is independent over slots and across users with $Pr[S_k(t)=0] = Pr[S_k(t)=1] = Pr[S_k(t)=2] = 1/3$.  We simulate over $10^6$ slots.   On slot $t=0$, we assign each user $k$ 
a desired file that is independently
in the other users with probability $p=0.05$, which establishes the $\script{F}_k$ sets.  These sets are held fixed
for the first third of the simulation, and then independently drawn again with a larger probability $p=0.1$ and 
held fixed for the second third.  New files are again drawn at the beginning of the final third of the 
simulation, with $p=0.07$.  

Fig. \ref{fig:samplepath-throughput}
plots the resulting throughput components from the base station traffic and 
peer-to-peer traffic separately, using $V=10$, $\alpha_k=\alpha = 0.5$, $\beta_k=0.05$, $x_k^{max}=3$. 
Even though there are only an average of 50/16 = 3.125 users per cell, 
and in the first third of the simulation there is only a $5\%$ chance that a given user has the file desired
by another user, the peer-to-peer traffic is still more than twice that of the base station alone.  This
further increases in the middle of the simulation when the file availability probability jumps to $10\%$.  
Fig. \ref{fig:samplepathQ} shows that the value of $Q_k(t)$ never exceeds 10 packets for any user $k$
(recall that the worst-case guarantee from Lemma \ref{lem:q-bound}
is $Q_k(t) \leq V+3 = 13$ packets).  All tit-for-tat constraints
were satisfied, with $H_k(t) \leq 24.6$ for all $k\in\script{K}$ and all $t$. 

Figs. \ref{fig:throughputV} and \ref{fig:QV} explore the throughput-backlog tradeoff with $V$ (one can also 
plot the throughput-utility with $V$ to see a similar convergence as in Fig. \ref{fig:throughputV}). 
Fig. \ref{fig:throughputV} also treats the case when the tit-for-tat constraint is made more stringent ($\alpha = 0.75$),
in which case throughput is reduced.  
There was no observed change in average queue backlog for $\alpha = 0.5$ and $\alpha = 0.75$ (the 
simulated curves look like one curve in Fig. \ref{fig:QV}. 

\begin{figure}[htbp]
   \centering
   \includegraphics[height=1.9in, width=3.5in]{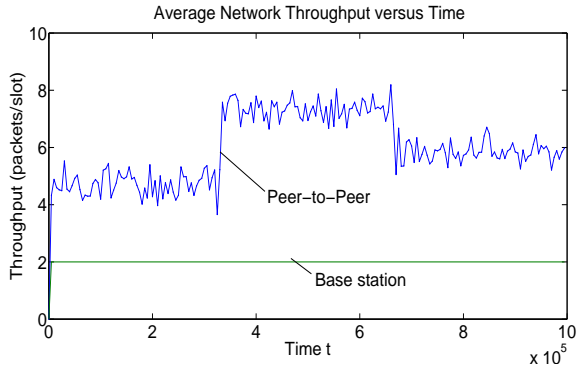} 
   \caption{Average throughput per-user versus time.}
   \label{fig:samplepath-throughput}
\end{figure}
 
 \begin{figure}[htbp]
   \centering
   \includegraphics[height=1.9in, width=3.5in]{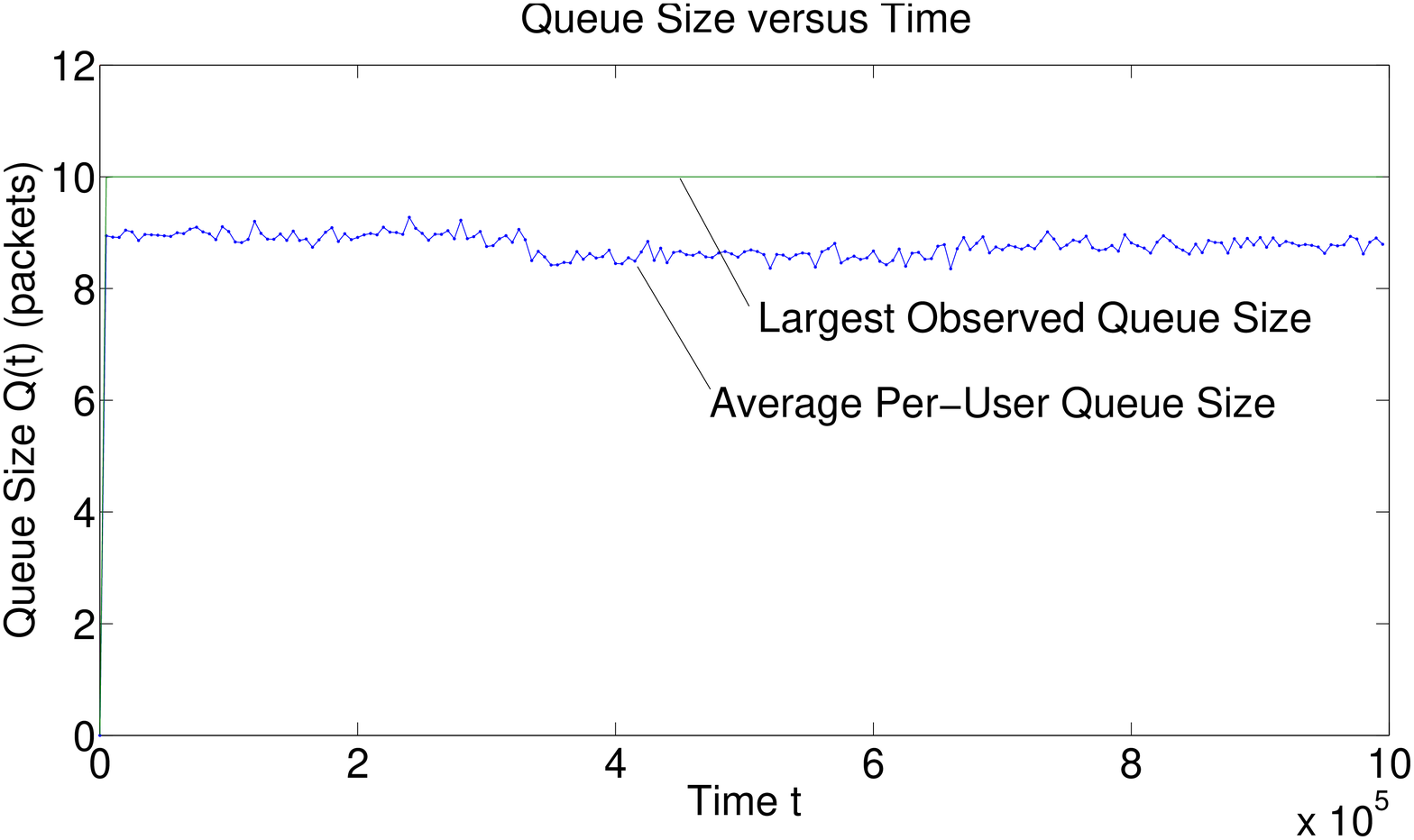} 
   \caption{Average and worst-case queue backlog per user versus time.}
   \label{fig:samplepathQ}
\end{figure}

\begin{figure}[htbp]
   \centering
   \includegraphics[height=1.9in, width=3.5in]{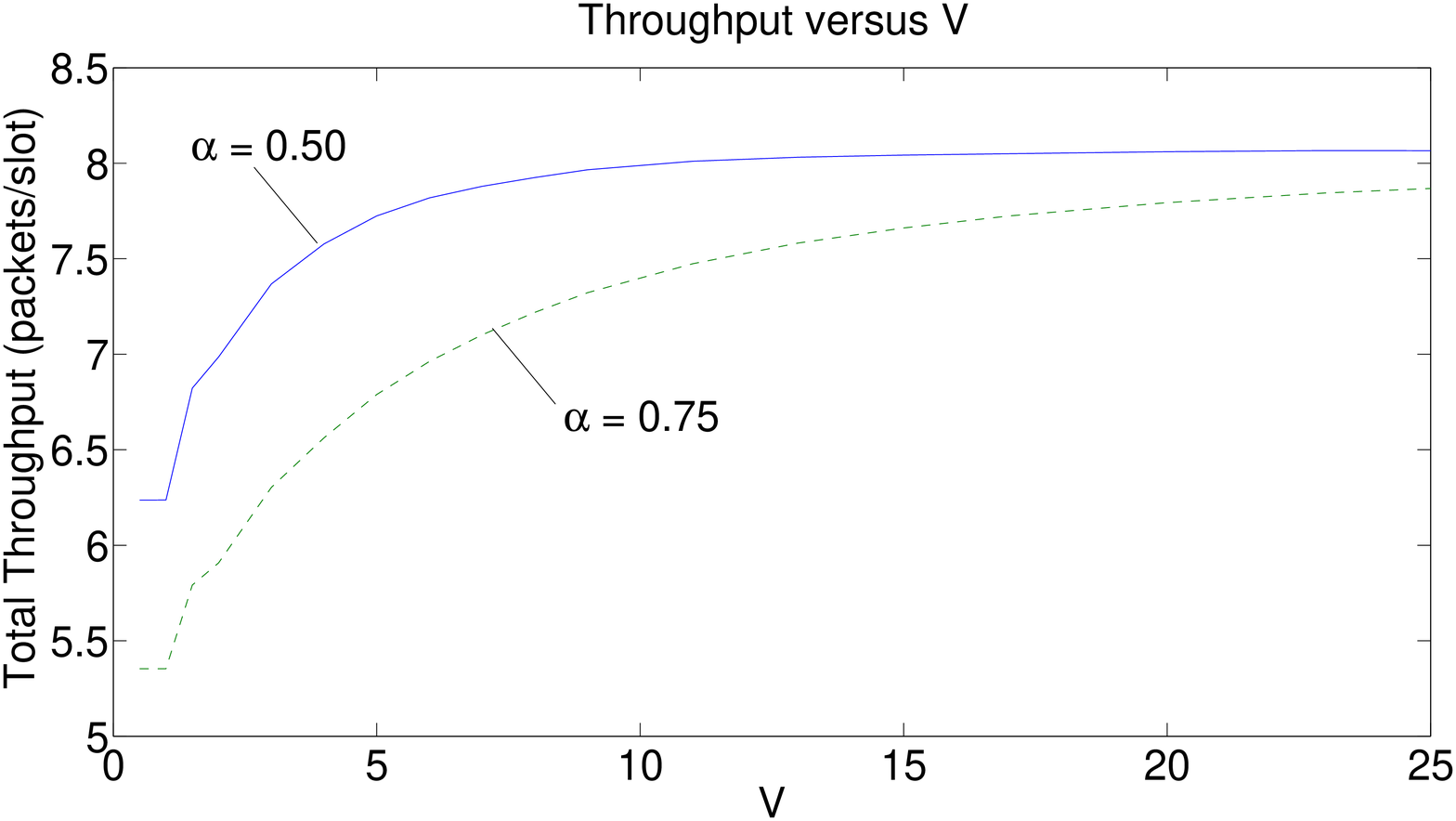} 
   \caption{Throughput versus $V$ for $\alpha = 0.5$ and $\alpha = 0.75$.}
   \label{fig:throughputV}
\end{figure}

\begin{figure}[htbp]
   \centering
   \includegraphics[height=1.9in, width=3.5in]{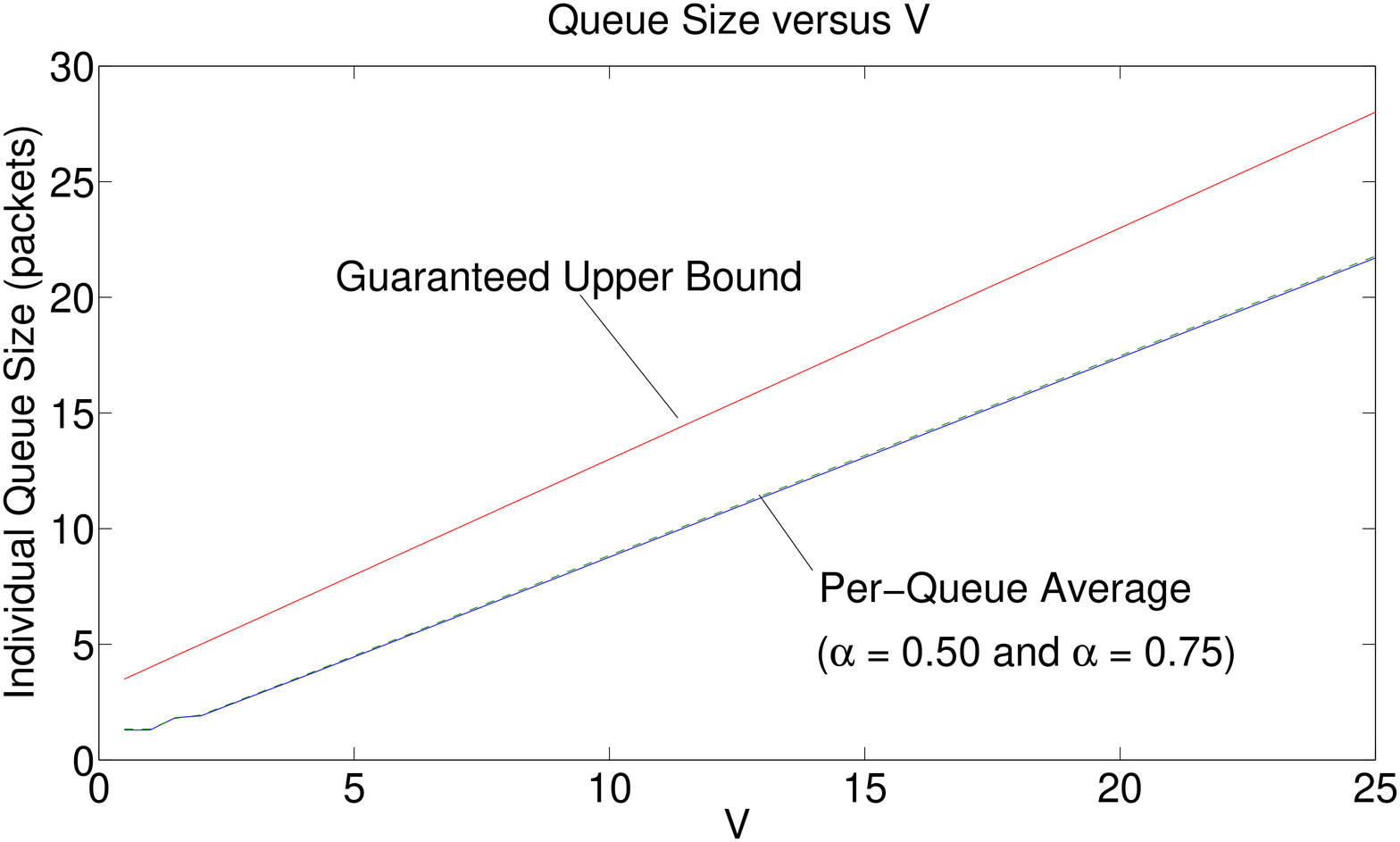} 
   \caption{Queue backlog versus $V$, demonstrating the $O(V)$ behavior.  
   The results for $\alpha=0.5$ and $\alpha=0.75$ are right on top
   of each other. The upper bound $V+3$ is also plotted.}
   \label{fig:QV}
\end{figure}

\section{Conclusions} 

This paper develops a simple peer-to-peer scheduling algorithm for a mobile wireless
network. Users can opportunistically grab packets from their peers who already have the desired
file in their cache.  This capability significantly increases the throughput capabilities in the network 
in comparison to a network where users can only download data from a base station.  Our simulations
demonstrate these gains, and illustrate that the algorithm is robust to non-ergodic events.  Analytically, 
we have shown that if each user requires a single file with infinite length, then the algorithm 
can achieve throughput 
utility that is arbitrarily close to optimality.  Distance to optimality depends on a $V$ parameter that 
also affects an $O(V)$ tradeoff in virtual queue sizes.  Our model embeds tit-for-tat constraints directly 
into the stochastic network optimization problem.  This incentivizes participation and naturally 
leads to an algorithm with \emph{reputation queues}.  Peer-to-peer links between 
users are favored according to the \emph{differential reputation} of the users.  Finally, while our 
results were developed for wireless scenarios, we note that 
the algorithm is general and the techniques can be used in other
contexts, such as in wireline computer networks and mixed wireless and wireline networks.

\section*{Appendix --- Proof of Lemma \ref{lem:h-bound}} 

Here we prove Lemma \ref{lem:h-bound}.  We state the lemma again for convenience: 

\emph{Lemma:} If all utility functions have right-derivatives bounded by 
 finite constants $\nu_k>0$, if $\beta_k>0$ for all $k \in \script{K}$, and if initial backlog satisfies
$Q_k(0) = H_k(0) =0$ for all $k \in \script{K}$, then there are finite constants $C_1$ and $C_2$,
both independent of $V$, such that: 
\begin{eqnarray*}
\norm{\bv{\Theta}(t)} \leq C_1 + C_2V \: \: \forall t \in \{0, 1, 2, \ldots\} 
\end{eqnarray*}
where $\norm{\bv{\Theta}(t)}$ is defined: 
\[ \norm{\bv{\Theta}(t)} \leq \sqrt{\sum_{k\in\script{K}}H_k(t)^2 + \sum_{k\in\script{K}}Q_k(t)^2} \]
 
 The constants $C_1$ and $C_2$ are explicitly computed in \eqref{eq:c1} and \eqref{eq:c2} below, 
 in terms of the $B$ constant from Lemma \ref{lem:dpp}. 
 
 \begin{proof}  
Because our algorithm makes decisions for $(\mu_{kb}(t))$ and $\gamma_k(t)$ 
that minimize the last three terms in the right-hand-side of \eqref{eq:dpp} over all alternative
feasible decisions, they give a value that is less than or equal to the value when the corresponding 
terms use the alternative feasible decisions $\tilde{\gamma}_k(t) = \tilde{\mu}_{kb}(t) = 0$.  Thus: 
\begin{eqnarray*}
\Delta(t) - V\sum_{k\in\script{K}}\phi_k(\gamma_k(t)) &\leq& B(t) - V\sum_{k\in\script{K}} \phi_k(0) \\
&& - \sum_{k\in\script{K}} H_k(t)\beta_k \\
&\leq& B - \beta_{min}\norm{\bv{H}(t)} 
\end{eqnarray*}
where $B$ is the constant that upper bounds $B(t)$ from Lemma \ref{lem:dpp} 
for all $t$, $\beta_{min} \defequiv \min_{k\in\script{K}}\beta_k$,
and we have used the fact that $\sum_{k\in\script{K}} H_k(t) \geq \norm{\bv{H}(t)}$.  Now define 
$C_0 \defequiv \sum_{k\in\script{K}}[\phi_k(x_k^{max})- \phi_k(0)]$.  We have for all $t$:
\[ \Delta(t) \leq B + VC_0 - \beta_{min}\norm{\bv{H}(t)} \]
Note that: 
\[ \Delta(t) = \frac{1}{2}\norm{\bv{\Theta}(t+1)}^2 - \frac{1}{2}\norm{\bv{\Theta}(t)}^2 \]
Thus: 
\begin{equation} \label{eq:follows}
 \norm{\bv{\Theta}(t+1)}^2 - \norm{\bv{\Theta}(t)}^2 \leq 2(B+VC_0) - 2\beta_{min}\norm{\bv{H}(t)} 
 \end{equation} 
Note by Lemma \ref{lem:q-bound} 
that for all slots $t$, we have $0 \leq Q_k(t) \leq Q_{max}$, where: 
\[ Q_{max} \defequiv \nu_{max}V + x_{max} \]
where $\nu_{max} \defequiv \max_{k\in\script{K}} \nu_k$ and $x_{max} \defequiv \max_{k\in\script{K}} x_k^{max}$. 
Thus, for any slot $t$: 
\begin{eqnarray}
\norm{\bv{\Theta}(t)} &\leq& \norm{\bv{Q}(t)} + \norm{\bv{H}(t)} \nonumber \\
&\leq& Q_{max}\sqrt{K} + \norm{\bv{H}(t)} \label{eq:combine} 
\end{eqnarray}
Now suppose that on slot $t$, we have: 
\begin{equation} \label{eq:suppose} 
 \norm{\bv{\Theta}(t)} > \frac{B + VC_0}{\beta_{min}} + Q_{max}\sqrt{K} 
 \end{equation} 
 Combining this with \eqref{eq:combine} shows that if \eqref{eq:suppose} holds, then: 
  \begin{equation} \label{eq:follows2} 
   \norm{\bv{H}(t)} > \frac{B + VC_0}{\beta_{min}} 
   \end{equation} 
It follows that if \eqref{eq:suppose} holds, then (by combining \eqref{eq:follows} and \eqref{eq:follows2}): 
\[ \norm{\bv{\Theta}(t+1)}^2 - \norm{\bv{\Theta}(t)}^2 < 0 \]
Thus, $\norm{\bv{\Theta}(t)}$ cannot increase if \eqref{eq:suppose} holds on slot $t$. 
Because the initial queue backlog is less than the threshold given in  \eqref{eq:suppose}, 
it follows that for all $t$: 
\[ \norm{\bv{\Theta}(t)} \leq \frac{B + VC_0}{\beta_{min}} + Q_{max}\sqrt{K} + g \]
where $g$ is defined as the maximum possible increase in $\norm{\bv{\Theta}(t)}$ in one slot.
Because both $Q_k(t)$ and $H_k(t)$ can increase by at most $x_{max}$ in one slot, we have
$g \leq x_{max}\sqrt{2K}$.
Thus, for all slots $t$ we have: 
\begin{eqnarray*}
\norm{\bv{\Theta}(t)} &\leq& \frac{B + VC_0}{\beta_{min}} + (V\nu_{max} + x_{max})\sqrt{K}  \\
&& + x_{max}\sqrt{2K} \\
&\leq& C_1 + C_2V 
\end{eqnarray*}
where: 
\begin{eqnarray}
C_1 &\defequiv& B/\beta_{min} + x_{max}(\sqrt{K} + \sqrt{2K})  \label{eq:c1} \\
C_2 &\defequiv& C_0/\beta_{min} + \nu_{max}\sqrt{K} \label{eq:c2} 
\end{eqnarray}
\end{proof} 
   
\bibliographystyle{unsrt}
\bibliography{../../latex-mit/bibliography/refs}
\end{document}